\documentclass[myspacing,11pt]{ectaart2}

\RequirePackage[OT1]{fontenc} \RequirePackage{amsthm}
\RequirePackage[cmex10]{amsmath} 
\RequirePackage{amsfonts} \RequirePackage{amssymb} \RequirePackage{a4wide}
\RequirePackage{graphicx}

\usepackage{xcolor}

\startlocaldefs \numberwithin{equation}{section}
\theoremstyle{plain}

\newtheorem{theorem}{Theorem}[section]

\newtheorem{remark}{Remark}[section]

\endlocaldefs

\def\@bysame#1{\vrule height 1.5pt depth -1pt width 3em \hskip
0.5em\relax}

\newcommand{\R}{ \mathbb{R} }

\newcommand{\wh}[1]{ \widehat{ #1 } }
\newcommand{\wt}[1]{ \widetilde{ #1 } }

\newcommand{\calI}{\mathcal{I}}

\newcommand{\calY}{\mathcal{Y}}

\newcommand{\Var}{{\mbox{Var\,}}}

\endlocaldefs

\begin{document}

\begin{frontmatter}

\title{On the Accuracy of Fixed Sampled and Fixed Width Confidence Intervals Based on the Vertically Weighted Average}
\runtitle{}


\begin{aug}
\author{\fnms{Ansgar} \snm{Steland${}^\ast$}
\ead[label=e1]{steland@stochastik.rwth-aachen.de}}
\ead[label=u1,url]{www.stochastik.rwth-aachen.de}


\address{Institute of Statistics\\ RWTH Aachen University \\ W\"ullnerstr. 3, D-52056 Aachen, Germany\\
\printead{e1}} 

\end{aug}


\runauthor{A. Steland}

\begin{abstract}
Vertically weighted averages perform a bilateral filtering of data, in order to preserve fine details of the underlying signal, especially discontinuities such as jumps (in dimension one) or edges (in dimension two). In homogeneous regions of the domain the procedure smoothes the data by averaging nearby data points to reduce the noise, whereas in inhomogenous regions the neighboring points are only taken into account when their value is close to the current one. This results in a denoised reconstruction or estimate of the true signal without blurring finer details. 

This paper addresses the lack of results about the construction and evaluation of confidence intervals based on the vertically weighted average, which is required for a proper statistical evaluation of its estimation accuracy. Based on recent results we discuss and investigate in greater detail fixed sample as well as fixed width (conditional) confidence intervals constructed from this estimator. The fixed width approach allows to specify explicitly the estimator's accuracy  and determines a random sample size to ensure the required coverage probability. This also fixes to some extent the inherent property of the vertically weighted average that its variability is higher in low-density regions than in high-density regions. To estimate the variances required to construct the procedures, we rely on resampling techniques, especially the bootstrap and the jackknife. 

Extensive Monte Carlo simulations show that, in general, the proposed confidence intervals are highly reliable in terms of their coverage probabilities for a wide range of parameter settings. The performance can be further increased by the bootstrap. 
\end{abstract}



\end{frontmatter}

\section{Introduction}

The analysis of noisy data, for instance obtained from sensors monitoring a signal, is a well studied but still challenging problem. This especially applies in applications such as brain imaging based on magnetic resonance imaging, where increasing the precision of the physical measurement system, which is given by squids making use of quantum effects to record weak magnetic fields, is very expensive. In other applications the preservation of discontinuities in terms of their heights and/or their locations matters. To reduce the noise one may apply a denoising procedure such as a low pass filter, but this often leads to a blurred signal corresponding to a substantial loss of information. Therefore, if discontinuities such as jumps are expected in the signal, one should apply jump-preserving procedures, which are often called edge-preserving in the literature due to their importance when processing two-dimensional image data. Various approaches have been studied in the literature such as wavelets as studied by \cite{Donoho1994}, point-wise adaptive approaches proposed by \cite{PolSpok2003}, kernel smoothing, see
\cite{WandJones1995} and \cite{Qiu2005} for its application to image analysis, 
or jump-preserving regression, \cite{Qiu1998}, to mention just a few works.

A large amount of smoothing statistics aiming at denoising can be (approximately) written as weighted averages of the observations $ Y_i $, $ i \in \calI = \{ 1, \dots, n \} $, attaining values in the range $ \calY \subset \R $, so that the definition of the weights attached to the data points $ (i, Y_i) \in \calI \times \calY $, $ i \in \calI $, of the corresponding scatterplot matters. In a univariate setting as studied here, the domain space $ \calI$ is usually equipped with the distance $ d_\calI(i,j) = |i-j| $ for $ i, j \in \calI $; in higher dimensions the Euclidean distance is the common choice. In order to preserve discontinuities, methods such as the bilateral filter, see \cite{TM1998}, use weights which depend on both the distance in the domain space $ \calI $ (horizontal weighting) and the range space $ \calY $ (vertical weighting), whereas, for example, local polynomial estimators use weights only depending on the distance in the domain space.

In this article, we focus on a specification of the vertical weighting approach where the weight attached to each observation is a function of its difference to the observation of interest, resulting in a vertically weighted average.
This approach has been extended to vertically weighted regressions for image analysis, \cite{PRS2008}, and jump-preserving monitoring procedures, see e.g. \cite{PRS2004}, \cite{PRS2010} and \cite{RS2009}. A related clipping median statistic has been investigated in \cite{Steland2004} for a general mixture model. In the present work, we confine our discussion to the univariate setting and study the problem how to evaluate the accuracy of the vertically weighted average in terms of confidence intervals. 

The computational costs of the vertically weighted average are of the order $ O(n) $, if $n$ denotes the sample size, which makes it attractive for realtime applications. Hence it is much faster to compute than other popular methods. For example, the computational costs of the fastest algorithm to calculate a $ \ell_0 $-penalized $M$-estimator are of the order $O(n^2)$, see \cite{FriedrichKempeLiebscherWinkler2008}. 
This of particular importance for large-scale applications to big data, e.g. large databases , say, of image data, where the procedure is applied thousands or even millions of times, first to apply to denoise the images and, second, to assess by resampling methods the precision of the resulting denoised reconstruction, e.g. in terms of a confidence interval. The version of the vertically weighted average studied in this paper has the advantage that it can be even implemented in hardware, such that it is well suited for real-time applications.

Basically, there are two approaches to construct a confidence interval. The more common approach is the fixed (but large) sample confidence interval, whose width is random, such that the precision indicated by the interval is random and cannot be specified. In order to set up a fixed width confidence interval, one needs to rely on a random sample size, which is determined in such a way that the confidence interval attains the preassigned coverage probability, asymptotically. We investigate a two stage approach. The first stage sample (small) is used to estimate the dispersion of the estimator on which the interval is based. Using that estimator one estimates the optimal sample size, which is therefore random. There is a rich literature on such two stage approaches to construct fixed width intervals, we refer to \cite{MS2009} and the references given there.

To estimate the variability of the vertically weighted average, we propose to rely on resampling techniques, since they are easy to apply and typically lead to convincing results. We study the jackknife variance estimator as a versatile and fast method, whose consistency for the vertically weighted average has been shown in \cite{StelandSQA2015}, and the bootstrap as a general tool. 

The question arises how the proposed confidence intervals work in practice. To address this issue, extensive simulations have been conducted to study the accuracy in terms of the coverage probablity. It turns out that the accuracy is, in general, very good, but the coverage may be lower than nominal in the tails of the distribution for the conditional approach.

The organisation of the paper is as follows: Section~\ref{Sec: Defs} introduces the vertically weighted average, establishes some properties justifying its interpretation as a signal estimator and discusses the fixed sample intervals. Section~\ref{Sec: Fixed-width CI} provides some further background on fixed width confidence intervals and introduces the proposed procedure in detail. The simulation studies are presented in Section~\ref{Sec: Simulations}.

\section{The vertically weighted average}
\label{Sec: Defs}

Let us assume for a moment that interest focuses on $ Y_n $, its mean $m_n = E(Y_n) $ and the relationship of $ m_n $ to the means $ m_i = E(Y_i) $, $ i = 1, \dots, m $ of the remaining sample of size $ m = n-1 $. To keep the presentation simple, we focus on this one-dimensional problem formulation and notice that it also covers a simplified version of the bilateral filter as studied in image processing: Here $ Y_i $, $ i = 1, \dots, n $, are the gray values of $ n $ pixels of a connected (usually rectangular) area of the image, where $ Y_n $ is the gray value of the center and $ Y_1, \dots, Y_{n-1} $ represent the gray values of the neighboring pixels. However, the spatial distance of the pixels to the center and other issues which are of some importance in image processing are not taken into account and should be addressed by future research.

Let $ Y_1, \dots, Y_n $ be independent real--valued observations following the model
\begin{equation}
\label{ModelOfPaper}
  Y_i = m_i + \epsilon_i, \qquad i = 1, \dots, n,
\end{equation}
where $ \epsilon_1, \dots, \epsilon_n $ are i.i.d. random variables with common distribution function $ F $, denoted by $ \epsilon_i $ i.i.d($F$), symmetrically distributed around $0$ and having finite second moment. $ m_i $, $ i = 1, \dots, n $, specifying the underlying true signal as $ m_i = E(Y_i) $, $ i = 1, \dots, n $. It is common to assume that the observed data are obtained by sampling equidistantly at time instants $ t_i = i \Delta $, $ i = 1, \dots, n $, an underlying function $ f : \mathcal{D} \to \R $, where $ \Delta > 0 $ the sampling period
(either constant or given by, say, $ \overline{\tau}/n $), and the domain is $ \mathcal{D} = [0, \infty) $ (if $ \Delta $ is fixed) or $ \mathcal{D} = [0,\overline{\tau}] $ (for $ \Delta = \overline{\tau}/n $), see e.g. \cite{PawlakSteland2013}. In this case $f$ is regarded as the true signal. The methods and results discussed in the present paper, however, do not need to assume this sampling model but are valid for the more general model (\ref{ModelOfPaper}).

Focusing on a neighborhood of the current observation $ Y_n $, we want to dampen the noise by some averaging procedure in such a way that large differences
in the means, $ m_i$, are preserved. This can be achieved by  averaging those data points $ Y_i $ whose values are close to the observation of interest, $ Y_n $. Therefore, the weights used to define a weighted mean should depend
on the differences $ Y_i - Y_n $ between the observations $ Y_i $  and $ Y_n $. It is natural to evaluate the difference $ Y_i - Y_n $ by means of a nonnegative and symmetric kernel function $ k : \mathbb{R} \to \mathbb{R} $ which is assumed to satisfy the conditions
\[
   E( k(\epsilon_1) ) < \infty \qquad \text{and} \qquad E( \| \epsilon_1 \|^2 k^r( \epsilon_1 ) ) < \infty, \ r = 1, 2.
\]
The {\em vertically weighted average} is now defined as 
\[
  \wh{\mu}_n = \wh{\mu}_n(Y_n) = \sum_{i \le m} Y_i k( Y_i - Y_n ) \ / \ \sum_{i \le m} k( Y_i - Y_n ),
  \qquad m = n-1,
\]
Typical kernels used in applications are the Gaussian kernel, i.e. the density of the $N(0,1) $ distribution, or kernels with support $ [-1,1] $ such as the uniform kernel $ 1_{[-1,1]} $. Further, often one incorporates a scale parameter $ \sigma > 0 $ and considers the choice
\[
  k(z) = \wt{k}(z/\sigma)
\]
for some generic kernel $ \wt{k} $. If $ \wt{k} $ is the uniform kernel, only those observations  $Y_i $ are averaged which are close to $ Y_n $ in the sense that $ |Y_i-Y_n| \le \sigma $. Figure~\ref{illustration} illustrates the denoising effect and the jump--preserving property of the approach.

\begin{center}
\begin{figure}
  \includegraphics[width=12cm]{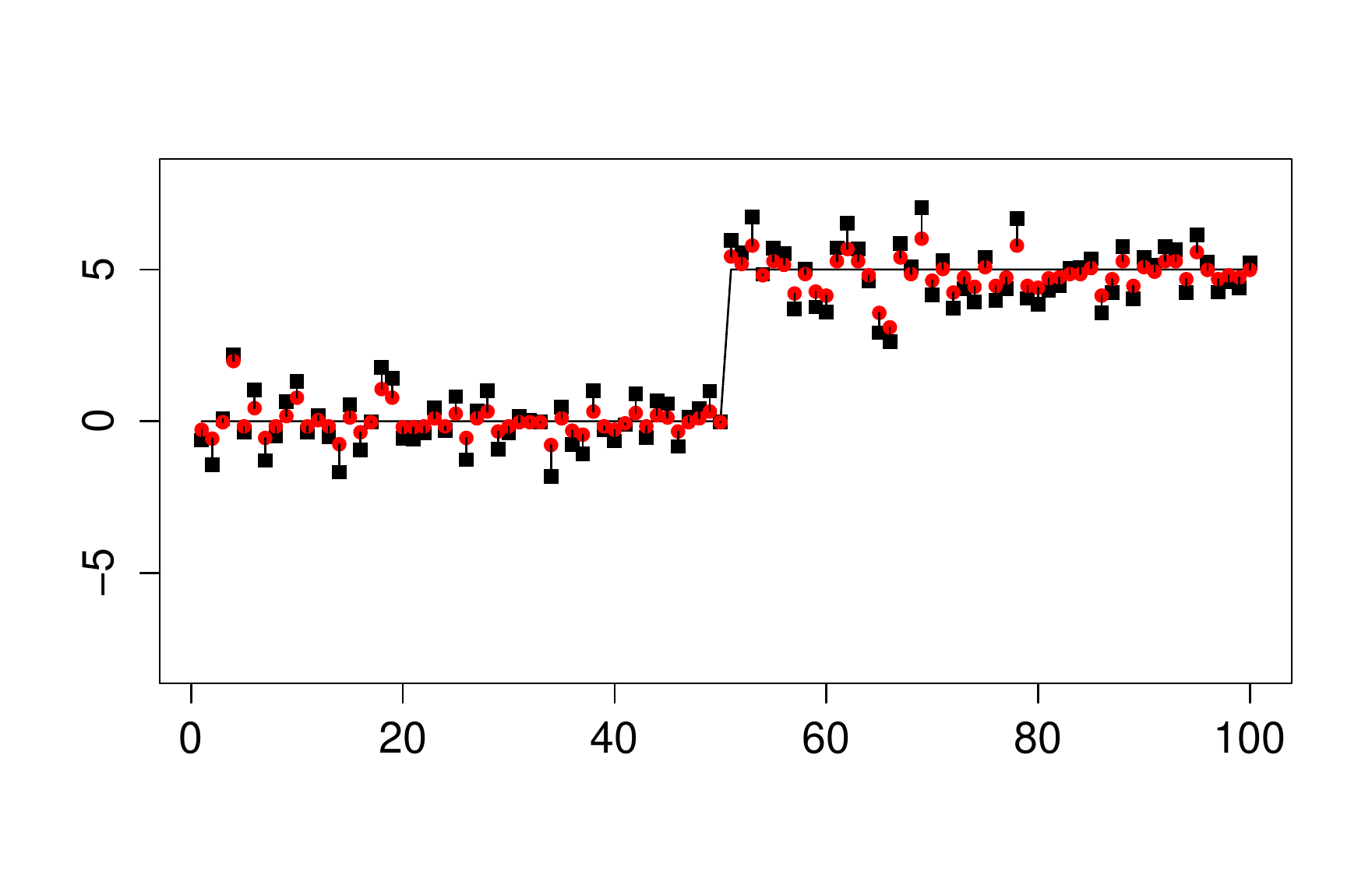}
  \caption{The vertically weighted average estimator for a noisy piece-wise signal using a uniform kernel and $ \sigma = 1.6 $. The observations are shrunken towards the true mean to obtain a denoised
  	estimate.}
  \label{illustration}
\end{figure}
\end{center}

In what follows, it will be sometimes convenient to write $ \wh{\mu}_n( Y_n; Y_1, \dots, Y_{n-1} ) $. In the same vain, the corresponding vertically weighted average associated to $ Y_i $ is given by
\[
  \wh{\mu}_n( Y_i ) = \wh{\mu}_n( Y_i; Y_1, \dots, Y_{i-1}, Y_{i+1}, \dots, Y_n ), 
\]
for $ i = 1, \dots, n $. Then $ \wh{\mu}_n(Y_1), \dots, \wh{\mu}_n(Y_n) $ is regarded as the denoised reconstructed signal.

The statistic $ \wh{\mu}_n $ is related to the vertically weighted mean squared functional
\[
  \tau(x;i) := E( | Y_i - x |^2 k(\epsilon_1) ), \qquad x \in \R,
\]
where $ i \in \{ 1, \dots, n \} $. As shown in \cite{PR2000}, see also \cite{Raf2007}, the functional $ \tau( \bullet; i) $ is minimized by the true mean $ m_i $ and $ m_i $ is a fix point of the nonlinear equation, i.e. a solution of the associated fix point equation
\[
  x = \frac{ E( Y_i k( Y_i - x ) ) }{ E k(Y_i - x) }.
\]
For extensions to Hilbert--valued random elements and further discussion see \cite{StelandSQA2015}. The statistic $ \wh{\mu}_n $ is obtained by replacing the expectations by their empirical sample analogs based on the sample $ \{ Y_1, \dots, Y_{n-1} \} $ and substituting the true mean, $m_n$, by the current observation $ Y_n $. 

The question arises how one may evaluate the accuracy of this statistic. A common approach is to consider confidence
intervals for the underlying parameter when the sample is homogeneous, i.e. if the null hypothesis of a constant signal,
\begin{equation}
\label{H0Cond}
  H_0:  m_1 = m_i, \ \text{for all $i \ge 2 $},
\end{equation}
holds true. This requires appropriate large sample asymptotics. But the statistic $ \wh{\mu}_n $ is not a member of standard classes of statistics, which hinders the application of known limit theorems to construct procedures for statistical inference, and therefore requires a special treatment. In \cite{StelandSQA2015} general invariance principles have been established which imply the following central limit theorem: If $ Y_1, Y_2, \dots $ are i.i.d. with existing fourth moment, then
\begin{equation}
\label{CLTCond}
  \sqrt{n} [ \wh{\mu}_n( Y_n ) - \theta( Y_n ) ] / \sigma_\xi( Y_n ) \stackrel{d}{\to} N(0,1),
\end{equation}
as $n \to \infty $, where $ \sigma_\xi^2(y) = \Var( \xi_1(y) ) $ with
\[ 
  \xi_i(y) = \frac{ k(Y_i-y) Y_i - \mu(y) }{ \nu(y) } + \theta(y) \frac{ k(Y_i-y) - \nu(y) }{ \nu(y) }, \qquad i = 1, \dots, n,
\]
and
\[
  \theta(y) = \frac{\mu(y)}{\nu(y)}, \quad \text{with} \quad  \mu(y) = E( k(Y-y)Y ) \text{\ and\ } \nu(y) = E( k(Y-y) ) > 0.
\]
for $ y \in \mathbb{R} $. The CLT (\ref{CLTCond}) holds true under the conditional law given $ Y_n $ and therefore also unconditionally.

What is the relation between the centering term in (\ref{CLTCond}), i.e. $ \theta(Y_n) $, and the true underlying signal, i.e. the constants in model (\ref{ModelOfPaper})? The following theorem, which is related to the results in \cite{PR2000} and \cite{Raf2007}, shows that the centering term vanishes in median and expectation, respectively, for i.i.d. samples.

\begin{theorem}
\label{LemmaTheta}
Assume that $ \epsilon_1, \epsilon_2 $ are i.i.d. with $ \epsilon_1 \stackrel{d}{=} - \epsilon_1, $ 
and $ Y_i = m + \epsilon_i $, $ i = 1, 2 $. If $k$ is a symmetric kernel, then 
\[
  \text{Med\,} \theta(Y_1) = m
\]
and $ E \theta(Y_1) = m $ if $ E \theta(Y_1) $ exists.
\end{theorem}

\begin{proof} The result can be shown using equal-in-distribution arguments. If $ X $ and $ Z$ are two random variables or random vectors of the same dimension with distributions $ P_X $ and $ P_Z $, we write $ X \stackrel{d}{=} Z $, if $P_X = P_Z $, i.e. if $ P( X \in A ) = P( Z \in A ) $ for all measurable sets $A$. If $ X \stackrel{d}{=} Z $, then $ h(X) \stackrel{d}{=} h(Z) $ for any measurable mapping $h$ which does not depend on $ (X,Z) $. Further, in what follows, we write $ E_X $ when the expectation is taken with respect to (the distribution of) $X$. Then, e.g., $ E_X h(X,Z) = \int h(x,Z) \, d P_X(x) $, if $ X $ and $Z$ are independent. We shall apply those basic results to the representation
\[
  \theta(Y_1) = \frac{ E_{Y_2} [ k(Y_2-Y_1) Y_2 ] }{ E_{Y_2} k(Y_2-Y_1) },
\]
where one has to take into account that numerator and denominator are depend. First notice that
\begin{align*}
  \text{Med\,} \theta(Y_1) & = \text{Med}_{Y_1} \left( \frac{ E_{Y_2} [ k(Y_2-Y_1) Y_2 ] }{ E_{Y_2} k(Y_2-Y_1) } \right) \\
  & = \text{Med}_{\epsilon_1} \left( \frac{ E_{\epsilon_2}[ k(\epsilon_2 - \epsilon_1) (\epsilon_2 + m) ] }{ E_{\epsilon_2} k(\epsilon_2 - \epsilon_1 ) } \right)  \\
  & = \text{Med}_{\varepsilon_1} \left( m + \frac{ E_{\epsilon_2} [ k(\epsilon_2-\epsilon_1) \epsilon_2 ] }{ E_{\epsilon_2} k(\epsilon_2 - \epsilon_1 ) } \right) \\
  & = m + \text{Med}_{\epsilon_1} \left( \frac{ g(\epsilon_1) }{ h(\epsilon_1) } \right)
\end{align*}
where
\begin{align*}
  g(\epsilon_1)  & = E_{\epsilon_2} [ k(\epsilon_2-\epsilon_1) \epsilon_2 ], \\
  h( \epsilon_1) & = E_{\epsilon_2} k(\epsilon_2 - \epsilon_1 ).
\end{align*}
We show that
\begin{equation}
\label{ToVerify}
  S( \epsilon_1 ) := ( g( \epsilon_1 ), h( \epsilon_1 ) ) \stackrel{d}{=} ( - g( \epsilon_1 ), h( \epsilon_1 ) ),
\end{equation}
which implies  
\[
  \text{Med}_{\epsilon_1} \left( \frac{ g( \epsilon_1 ) }{ h( \epsilon_1 ) } \right) 
  = -  \text{Med}_{\epsilon_1} \left( \frac{ g( \epsilon_1 ) }{ h( \epsilon_1 ) } \right),
\]
thus showing $ \text{Med\,} \theta(Y_1) = m $. The corresponding property for the expectation follows now easily. To verify (\ref{ToVerify}) observe that
\begin{align*}
  S( \epsilon_1) & = ( E_{\epsilon_2}[ k(\epsilon_1-\epsilon_2) \epsilon_2 ], E_{\epsilon_2} k(\epsilon_1-\epsilon_2) )    \\
    & = (- E_{\epsilon_2} [ k( \epsilon_1 + \epsilon_2 ) \epsilon_2 ], E_{\epsilon_2} k( \epsilon_1 + \epsilon_2 ) )
\end{align*}
since $ \epsilon_1 \stackrel{d}{=} - \epsilon_1 $ and $ k(-z) = k(z) $. Next, using $ \epsilon_2 \stackrel{d}{=} - \epsilon_2 $,
we obtain
\[
  S( \epsilon_1 ) \stackrel{d}{=} \widetilde{S}(\epsilon_1)= (- E_{\epsilon_2}[ k(\epsilon_1-\epsilon_2) \epsilon_2 ], E_{\epsilon_2} k(\epsilon_1 - \epsilon_2 ) ). 
\]
But $ k(z) = k(-z) $ implies 
\[
  \widetilde{S}(\epsilon_1) = ( - E_{\epsilon_2}[ k( \epsilon_2 - \epsilon_1 ) \epsilon_2 ], E_{\epsilon_2} k(\epsilon_2 - \epsilon_1 ) ),
\]
which completes the proof.
\end{proof}

To the best of the author's knowledge, the following result also does not appear in the literature.

\begin{theorem}
  Let $ Y_1, \dots, Y_n $ be i.i.d. following the model $ Y_i = m + \epsilon_i $, $ i = 1, \dots, n $, for some $ m \in \R $ with
  symmetric error terms $ \epsilon_1, \dots, \epsilon_n $. If $ k $ is symmetric, then
  \[
    \text{Med}( \wh{\mu}_n(Y_i) ) = m
  \]
  for $ i = 1, \dots, n $.
\end{theorem}

\begin{proof}
  It suffices to show the result for $ \wh{\mu}_n(Y_n) $. Notice that
  \[
    \wh{\mu}_n(Y_n) = m + \frac{ \sum_{i \le m} k( \epsilon_i - \epsilon_n ) \epsilon_i }{ \sum_{i \le m} k( \epsilon_i  - \epsilon_n ) }.
  \]
  By symmetry and independence we have $ ( \epsilon_1, \dots, \epsilon_n ) \stackrel{d}{=} - (\epsilon_1, \dots, \epsilon_n) $. 
  Hence we obtain
  \[
    \left( \sum_{i \le m} k( \epsilon_i - \epsilon_n ) \epsilon_i, \sum_{i\le m} k( \epsilon_i - \epsilon_n ) \right)
    \stackrel{d}{=}
    \left( - \sum_{i \le m} k( \epsilon_i - \epsilon_n ) \epsilon_i, \sum_{i\le m} k( \epsilon_i - \epsilon_n ) \right)    
  \]
  Therefore
  \[
    \frac{ \sum_{i \le m} k( \epsilon_i - \epsilon_n ) \epsilon_i  }{ \sum_{i\le m} k( \epsilon_i - \epsilon_n )  }
    \stackrel{d}{=}
    -\frac{ \sum_{i \le m} k( \epsilon_i - \epsilon_n ) \epsilon_i  }{ \sum_{i\le m} k( \epsilon_i - \epsilon_n )  }.
  \]
  But this implies $ \text{Med}\left( \frac{ \sum_{i \le m} k( \epsilon_i - \epsilon_n ) \epsilon_i  }{ \sum_{i\le m} k( \epsilon_i - \epsilon_n )  } \right) = 0 $. 
\end{proof}

\begin{remark} Note that the existence of  $ E \theta(Y_1) $ and 
$ E \wh{\mu}_n(Y_n) $ depends on the kernel $k$ and the error distribution. But it can be easily guaranteed by adding a positive but arbitrarily small constant to the kernel $ k $.
\end{remark}

Although the validity of the central limit theorem greatly eases the interpretation of an estimator and its variation, the above formulas are cumbersome to construct practical procedures. A simple and effective approach to estimate the variance of an estimator is Efron's jackknife, which dates back to the works of \cite{Que1949} and \cite{Tukey1958}. For the vertically weighted average $ \wh{\mu}_n(Y_n) $ it is given by
\[
  \wh{\sigma}_n^2(Y_n) = \frac{m-1}{m} \sum_{i=1}^m \left( \wh{\mu}_{n,-i}(Y_n) - \overline{\wh{\mu}}_n( Y_n ) \right)^2,
\]
where 
\[
  \wh{\mu}_{n,-i}(Y_n) = \sum_{j = 1 \atop j \not=i}^{m-1} k( Y_j - Y_n ) Y_j \ / \ \sum_{j=1 \atop j\not= i}^{m-1} k( Y_j - Y_n ),
\]
are the leave--one--out estimates, $ i = 1, \dots, m $, and $ \overline{\wh{\mu}}_n( Y_n ) = m^{-1} \sum_{j=1}^m \wh{\mu}_{n,-i}(Y_n) $. The jackknife variance estimator is consistent if (\ref{H0Cond}) holds and the r.v.s have a finite fourth moment, see \cite{StelandSQA2015}.

Having a consistent estimator for $ \text{sd}( \wh{\mu}_n(Y_n) ) = \sqrt{ \Var( \wh{\mu}_n(Y_n) \, |\, Y_n ) } $ at our disposal, one may calculate fixed sample asymptotic confidence intervals for $ \mu_n^{(0)}(Y_n) = E( \wh{\mu}_n | Y_n ) $, namely
\begin{equation}
\label{CIFixedSample}
  [ \wh{\mu}_n(Y_n) - \Phi^{-1}( 1- \alpha/2 ) \wh{\sigma}_n(Y_n), \wh{\mu}_n(Y_n) + \Phi^{-1}(1-\alpha/2 ) \wh{\sigma}_n(Y_n) ],
\end{equation}
for $ \alpha \in (0,1) $, in order to asses the estimator's precision given $Y_n$. In (\ref{CIFixedSample}) and in what follows, $ \Phi(z) = (2\pi)^{-1/2} \int_{-\infty}^z e^{-z^2/2} \, dt $, $ z \in \mathbb{R} $, is the distribution function of the standard normal distribution and $ \Phi^{-1} $ its quantile function. Below we shall discuss how one can determine the sample size to obtain uniform accuracy in terms of the width of the interval, for any $Y_n$. The coverage probability of those confidence intervals are investigated in Section~\ref{Sec: Simulations}.

The unconditional asymptotic normality of $ \wh{\mu}_n(Y_n) $ under the null hypothesis, which follows from the conditional central limit theorem, combined with Lemma~\ref{LemmaTheta} suggests that
\[
  \frac{ \wh{\mu}_n }{ \wh{sd}( \wh{\mu}_n ) } \stackrel{d}{\to} N(0,1),
\]
as $ n \to \infty $,  for any consistent estimator $ \wh{sd}( \wh{\mu}_n ) $ of $ \sqrt{ \Var( \wh{\mu}_n ) } $. Hence, we shall investigate in Section~\ref{Sec: Simulations} in greater detail the coverage probabilities of the bootstrap confidence interval 
\begin{equation}
\label{CIMarginal}
  \left[ \wh{\mu}_n  - \Phi^{-1}(1-\alpha/2) \wh{sd}( \wh{\mu}_n ),
    \wh{\mu}_n + \Phi^{-1}(1-\alpha/2) \wh{sd}( \wh{\mu}_n ) \right]
\end{equation}
where $ \wh{sd}^2( \wh{\mu}_n ) $ is the bootstrap variance of $ \wh{\mu}_n $ estimated from $ B $
replications $ \wh{\mu}_n^*(b) $, $ b = 1, \dots, B $, 
\[
  \wh{\mu}_n^*(b) = \sum_{i \le m} Y_k^*(b) k( Y_i^*(b) - Y_n^*(b) ) \ / \ \sum_{i \le m} k( Y_i^*(b) - Y_n^*(b) ),
\]
where $ Y_1^*(b), \dots, Y_n^*(b) $ are $ \text{i.i.d.}( \wh{F}_n ) $ with $ \wh{F}_n(x) = \frac{1}{n} \sum_{i=1}^n 1( Y_i \le x ) $, $ x \in \mathbb{R} $.


\section{Fixed--width confidence intervals}
\label{Sec: Fixed-width CI}

By construction and due to the very basic idea of the vertical weighting approach, the (conditional) variance of the estimator $ \wh{\mu}_n(Y_n) $ strongly depends on the value $ Y_n $. This can be easily seen if $ k(z) = 1_{[-1,1]}(z/h) $ for some fixed $ h > 0 $. Then $ \wh{\mu}_n(Y_n) $ is the sample mean of all $ Y_i $ with $ |Y_i-Y_n| \le h $. The effective sample size $ \#( i \in \{ 1, \dots, m \} : |Y_i-Y_n| \le h ) $ is random and typically large, if $ Y_n $ is located in the center of the distribution, but it will be small if $ Y_n $ is in the tail. 

The following approach to construct a fixed width confidence interval
for the vertically weighted average, introduced in \cite{StelandSQA2015},  overcomes that drawback and allows to determine a sample size that leads to a simple and sound interpretation, namely that the resulting estimator has a specified precision.  The basic idea is to determine the sample size in such a way that, for a preassigned accuracy $ d > 0 $, the two--sided fixed width confidence interval 
\[
 I_n(d) = [ \wh{\mu}_n(Y_n) - d, \wh{\mu}_n(Y_n) + d ] 
\] 
has {\em conditional asymptotic coverage} $ 1-\alpha $, $ \alpha \in (0,1/2) $ given, i.e.
\begin{equation}
\label{CondFWCI}
  P\bigl( \wh{\mu}_n(Y_n) - d, \wh{\mu}_n(Y_n) + d ] \ni \theta(Y_n) | Y_n \bigr) = 1-\alpha + o(1),
\end{equation}
as $ n \to \infty $, a.s. If the conditional asymptotic distribution of $ \wh{\mu}_n(Y_n) $ is normal and were completely known, one could determine the asymptotically optimal sample size $ n_{opt} $ required to ensure (\ref{CondFWCI}). But this fails in practice, since the asymptotic variance is unknown. In a two-stage procedure
one starts with a deterministic initial sample size $ n_0 $, depending on the precision parameter $d$ and the confidence level $ 1- \alpha $, and uses that initial sample to estimate the asymptotically optimal sample size $ n_{opt} $ required to achieve a fixed width confidence interval of length $2d $ with asymptotic coverage $ 1-\alpha $. As the sample size is estimated using the first--stage sample, the final sample size $N$ is random. 

The two-stage
procedure studied in \cite{StelandSQA2015} adopts the two-stage procedure from \cite{Muk1980} and works as follows: Let us fix the current observation and denote it by $ Y_N $, although the sample size is not yet determined. The current observation will be always the last one and the additional $N-1$ observations then correspond to the first $N-1$ data points in the sample.
We shall first determine an initial sample size $ n_0 $ for the first stage and set up the sample $ Y_1, \dots, Y_{n_0-1}, Y_N $ using $ n_0-1 $ observations in addition to the current one. Then the final sample size $N$ (or equivalently $M := N-1$) will be determined and, in the same way, a sample of size $N$ will be set up with the current observation $ Y_N $ put at the end of the sample. This is necessary because the formula for the vertically weighted average regards the last observation as the current one. At the first stage, one  draws a first-stage (initial) sample of size  $n_0 $ given by
\begin{equation}
\label{N0s}
  n_0 = n_0(d) = \max \{ \lfloor \Phi^{-1}(1-\alpha/2) / d \rfloor, 3 \}.
\end{equation}
At the second stage calculate the random final sample size
\begin{equation}
\label{N2st}
  N = N(d) = \max\{ n_0, \lfloor \wt{\sigma}_{n_0}^2(Y_N) \Phi^{-1}(1-\alpha/2)^2  / d^2 + 2\rfloor \},
\end{equation}
where
\begin{equation}
\label{JackSDas}
  \wt{\sigma}_{n_0}^2(Y_N) = (n_0-1) \wh{\sigma}_{n_0}^2(Y_N)
  = (n_0-1) \sum_{i=1}^{n_0-1} \left( \wh{\mu}_{n_0,-i}(Y_N) - \overline{\wh{\mu}}_{n_0}(Y_N) \right)^2
\end{equation}
with $ \overline{\wh{\mu}}_{n_0}(Y_N) = \frac{1}{n_0-1} \sum_{i=1}^{n_0-1} \wh{\mu}_{n_0,-i}(Y_N) $
is the jackknife estimator of the asymptotic variance of the vertically weighted average calculated from the first-stage initial sample of size $ n_0 $ (augmented by $Y_N$)); notice that we have $ n_0-1 $ leave-one-out estimates ($Y_N$ is fixed) leading to the formula (\ref{JackSDas}). This means, if $ N > n_0 $, we sample additional $ M-n_0 $ observations, where $ M = M(d) = N(d)-1$, to obtain the final sample $ Y_1, \dots, Y_{n_0}, Y_{n_0+1}, \dots, Y_M, Y_N $. 

For theoretical investigations interest focuses on the behavior of the proposed procedure when the precision parameter $d$, and thus the length of the confidence interval, tends to $0$. Notice that, by (\ref{N0s}), this implies $ n_0 \to \infty $, which in turn ensures that $ N \to \infty$, see (\ref{N2st}). Since $ n_{opt} $ as well as $ N $ tend to $ \infty $, as $ d \to 0 $, comparisons are based on their ratio, in order to establish the statistical properties of the sequence $ I_n(d), d > 0 $, of fixed-width confidence intervals. In \cite{StelandSQA2015} it has been shown that the resulting confidence interval $ I_n(d) $ has the following properties, when (\ref{H0Cond}) holds:

\begin{itemize}
\item[(i)] $I_N(d) $ has asymptotic coverage $ 1- \alpha $, as $ d \to 0 $.
\item[(ii)] $ I_N(d) $ is consistent for the asymptotically optimal fixed sample interval using the asymptotic optimal sample size $n_{opt} $, i.e. $ N(d) / n_{opt}  = 1 + o_P(1) $, as $ d \to 0 $.
\item[(iii)] $ I_N(d) $ is first-order asymptotic efficient in the sense of Chow and Robbins, i.e. \[ E(N|Y_N) / n_{opt} \to 1, \] as $ d \to 0 $.
\end{itemize}


Our simulations reported below show that the coverage probability of the two--stage fixed--width confidence interval is very good.  Its construction is based on a central limit theorem, which suggests to use a resampling approach such as the nonparametric bootstrap, in order to improve the approximation. More specifically, a closer look at the proof behind the validity of the two--stage procedure discussed above reveals that it is based on the central limit theorem
\begin{equation}
\label{ACLT}
  P\left( \sqrt{M} \frac{ \wh{\mu}_N(Y_N) - \theta(Y_N) }{ \wt{\sigma}_{n_0}(Y_N) } \le z \right) \to \Phi(z),
\end{equation}
as $ d \to 0 $. The sample size $ N = M + 1$ is then determined such that
\begin{equation}
\label{ACLT2}
  \Phi\left( \sqrt{M} \frac{d}{ \wt{\sigma}_{n_0}(Y_N) } \right) \stackrel{!}{=} 1 -\alpha/2,
\end{equation}
leading to the formula for $n$ given in the previous section. (\ref{ACLT}) and (\ref{ACLT2}) suggest to bootstrap the distribution of the statistic $ \wh{\mu}_N  := \sqrt{M} \frac{ \wh{\mu}_N(Y_N) - \theta(Y_N) }{ \wt{\sigma}_{n_0}(Y_N) } $ to improve upon the central limit theorem. This is achieved by substituting the bootstrap distribution estimator for the distribution function $ \Phi $ in (\ref{ACLT2}).

The nonparametric simulation-based bootstrap, which estimates the bootstrap distribution by a simulation based on $B$ replications, was applied to the problem of interest as follows: Draw $B$ independent resamples of size $ m^* = n^* - 1 $, the bootstrap sample size, from $ \wh{F}_{n_0} $,
\[ 
  Y_1^*(b), \dots, Y_{m^*}^*(b) \stackrel{i.i.d.}{\sim} \wh{F}_{n_0}, \qquad b = 1, \dots, B,
\] 
where $ \wh{F}_{n_0}(y) = n_0^{-1} \sum_{i=1}^{n_0} 1( Y_i \le y ) $, $ y \in \mathbb{R} $, is the empirical distribution function of the initial sample. Alternatively, one can use the smooth bootstrap which convolves $ \wh{F}_{n_0} $ with a Gaussian law, $ N( 0, \wh{h}^2 ) $. Here one may use the cross validated bandwidth selector for the kernel density estimator or the asymptotically optimal choice, $ 1.06 s n_0^{-1/5} $, for Gaussian data, where $s^2$ is the sample variance. For simplicity of presentation, let us denote the bootstrap observations by $ Y_1^*(b), \dots, Y_{m^*}^*(b) $, whatever kind of bootstrap has been used. Then form the bootstrap samples 
\[
   (Y_1^*(b), \dots, Y_{m^*}^*(b),Y_N), \qquad b = 1, \dots, B,
\] 
where $ Y_N $ is again the current observation, and calculate
\begin{equation}
\label{BootstrapReplicates}
  \wh{t}_{n^*,b} = \sqrt{m^*} \frac{ \wh{\mu}_{n^*,b} - \wh{\mu}_{n_0} }{ \wt{\sigma}_{n_0} }, \qquad b = 1, \dots, B,
\end{equation}
where $ \wh{\mu}_{n_0} = \wh{\mu}_{n_0}( Y_N ) $ and
\[ 
  \wh{\mu}_{n^*,b}  = \wh{\mu}_{n^*}( Y_1^*(b), \dots, Y_{m^*}^*(b),Y_N )
  = \sum_{i \le m^*} Y_i^*(b) k( Y_i^*(b) - Y_N ) \ / \ \sum_{i \le m} k( Y_i^*(b) - Y_N ),
\]
for $ b = 1, \dots, B $.
Lastly, estimate the $ (1-\alpha/2) $--quantile of the distribution of $ \sqrt{m^*} \frac{ \wh{\mu}_{n^*}(Y_N) - \theta(Y_N) }{ \wt{\sigma}_{n_0}(Y_N) }  $  by the corresponding order statistic $ \wh{t}_{m^*,( B(1-\alpha/2) )}^* $, where $ \wh{t}_{m^*,(1)}^* \le \cdots \le \wh{t}_{m^*,(B)}^* $ denotes the order statistic of the bootstrap replicates (\ref{BootstrapReplicates}). This leads to the bootstrap final sample size 
\begin{equation}
\label{NBoot}
  N^* = N^*(d) = \max\{ n_0, \lfloor \wt{\sigma}_{n_0}^2(Y_N) ^2 (\wh{\mu}_{m^*,( B(1-\alpha/2) )}^*)^2 / d^2 + 2\rfloor \}.
\end{equation}
To conduct this bootstrap procedure, one needs to select the bootstrap sample size $n^* $ resp. $m^* = n^*-1 $. In the simulation study $ n^* = \min( 1.5 n_0, 50) $ was used.

\section{Simulations}
\label{Sec: Simulations}

The simulations aim at studying the dispersion of the vertically weighted average and, especially, the performance of the proposed methods in terms of the coverage probability. Let us start with a first experiment to examine how the dispersion of the vertically weighted average depends on the location of the current observation $ Y_n $. Figure~\ref{SDest} depicts the standard deviation given $ Y_n = y $, $ y \in \mathcal{Y} = [\Phi^{-1}(0.05), \Phi^{-1}(0.95)] $, for the sample size $ n = 30 $ when using a Gaussian kernel with standard deviation $ \sigma = 0.4 $. The interval $ \mathcal{Y} $ was discretized using a step size of $ 0.025 $ and each resulting case was simulated using $ 100,000 $ repetitions. It can be seen that the dispersion is substantially larger in the tails than in the center of the distribution.

Those simulations as well as all studies presented in the following subsections are based on observations following a standard normal distribution, as the focus of the simulations is to investigate the performance for the three different confidence intervals (classical fixed sample, fixed width with CLT asymptotics and fixed width with bootstrap) when varying the method parameter $\sigma$, the confidence level $ 1-\alpha $ and the sample size (for fixed sample intervals) and the precision parameter $d$ (for the fixed width intervals), respectively.

\begin{figure}
\begin{center}
  \includegraphics[scale=0.6]{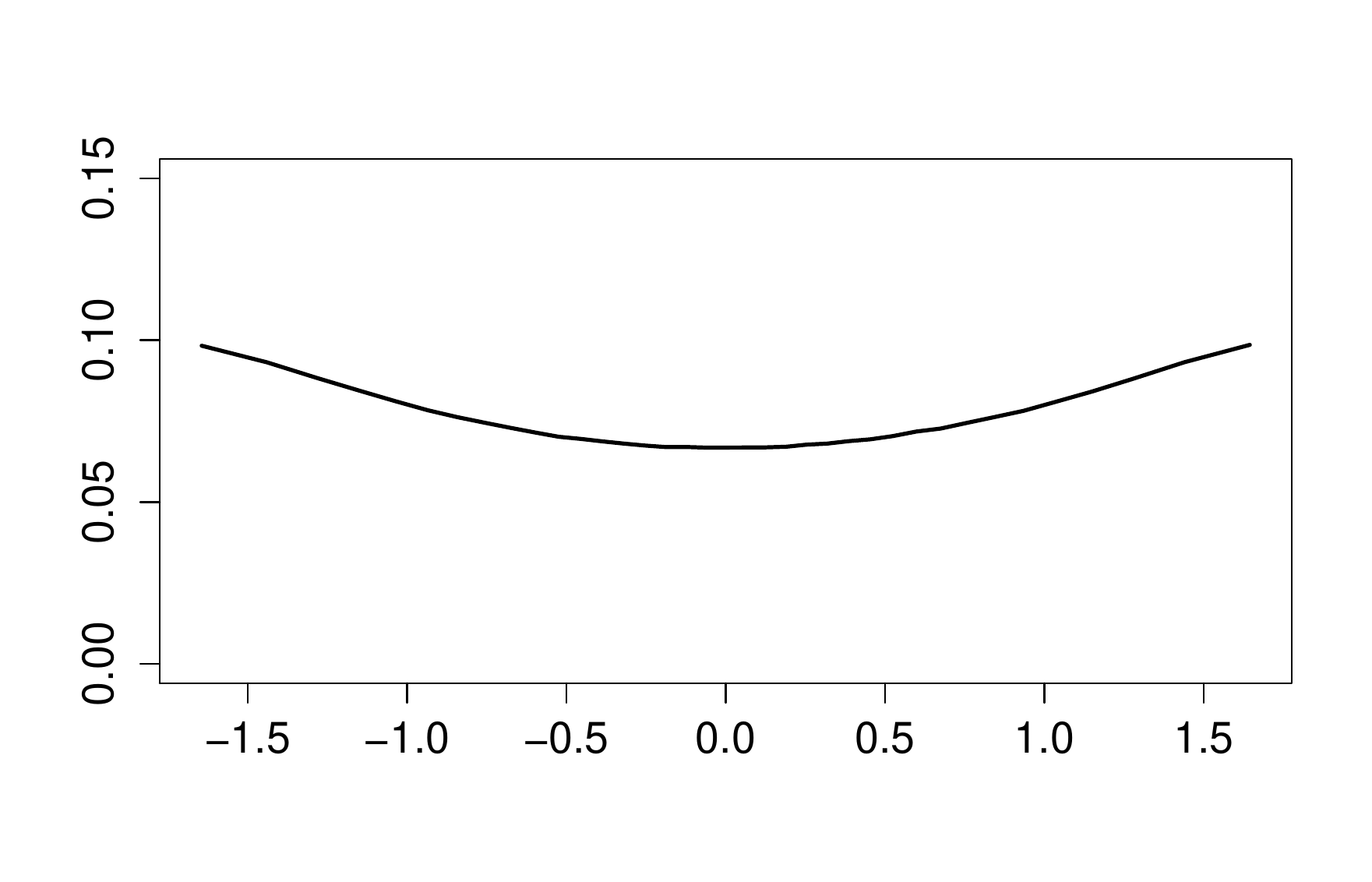}
\end{center}
\caption{The standard devation of $ \wh{\mu}_n(y) $ for $ y \in [\Phi^{-1}(0.05), \Phi^{-1}(0.95) ] $ under $ N(0,1) $ errors.}
\label{SDest}
\end{figure}

\subsection{Accuracy of conditional fixed sample confidence intervals}

Let us start our investigation by studying the coverage probability of the proposed asymptotic fixed sample confidence intervals (\ref{CIFixedSample}). The parameter $\sigma $ has levels $ 0.4, 0.6, 0.8 $ and the coverage probability was simulated across the distribution of $ Y $ by
conditioning on $ Y_n = y $ for $ y \in \{ \Phi^{-1}(q) : q \in \{ 0.05, 0.1, 0.3, 0.5, 0.8, 0.95 \} \} $. Additionally, the
coverage probability of the confidence interval as used in practice (real CI), i.e. for $ y=Y_n $ (randomly drawn) was considered.
The sample sizes under investigation are $ n = 20, 30, 50 $. Each case was simulated using $50,000 $ runs. The required true value $ E \wh{\mu}_n(y) $ was simulated based on a random sample of size $ m = 500,000 $. 

The results for a confidence level of $ 95\% $ are shown in Table~\ref{SimTableClassi}, for  $q$-values between $ 0.05 $ and $0.95$. It can be seen that the coverage probabilities are fairly good in the center of the distribution, but decrease in the tails, especially for small sample sizes such as $ n = 20 $. This results in a lower than nominal coverage of the real CI.

\begin{table}
 \centering\ 
 \begin{tabular}{ccrrrrrrrr} 
 \hline
 $\sigma$ & $n$ & \multicolumn{8}{c}{$q$} \\ 
   &   & $ 0.05 $& $ 0.1 $& $ 0.3 $& $ 0.5 $& $ 0.8 $& $ 0.9 $& $ 0.95 $& Real CI \\ 
 \hline
$ 0.4 $ & $ 20 $  & $  0.862  $  & $  0.873  $  & $  0.926  $  & $  0.934  $  & $  0.916  $  & $  0.873  $  & $  0.863  $  & $  0.891  $  \\ 
 & $ 30 $  & $  0.868  $  & $  0.904  $  & $  0.938  $  & $  0.941  $  & $  0.933  $  & $  0.907  $  & $  0.869  $  & $  0.911  $  \\ 
 & $ 50 $  & $  0.902  $  & $  0.931  $  & $  0.943  $  & $  0.944  $  & $  0.939  $  & $  0.931  $  & $  0.901  $  & $  0.925  $  \\ 
 $ 0.6 $ & $ 20 $  & $  0.853  $  & $  0.893  $  & $  0.934  $  & $  0.936  $  & $  0.924  $  & $  0.895  $  & $  0.854  $  & $  0.898  $  \\ 
 & $ 30 $  & $  0.878  $  & $  0.921  $  & $  0.941  $  & $  0.944  $  & $  0.935  $  & $  0.919  $  & $  0.880  $  & $  0.915  $  \\ 
 & $ 50 $  & $  0.917  $  & $  0.933  $  & $  0.944  $  & $  0.947  $  & $  0.942  $  & $  0.935  $  & $  0.919  $  & $  0.928  $  \\ 
 $ 0.8 $ & $ 20 $  & $  0.863  $  & $  0.909  $  & $  0.934  $  & $  0.936  $  & $  0.927  $  & $  0.906  $  & $  0.863  $  & $  0.902  $  \\ 
 & $ 30 $  & $  0.899  $  & $  0.925  $  & $  0.941  $  & $  0.942  $  & $  0.935  $  & $  0.926  $  & $  0.900  $  & $  0.918  $  \\ 
 & $ 50 $  & $  0.922  $  & $  0.936  $  & $  0.946  $  & $  0.948  $  & $  0.945  $  & $  0.937  $  & $  0.925  $  & $  0.932  $  \\ 
 \hline
\end{tabular}

 \ \\[0.1cm]
\caption{Coverage of classical fixed sample confidence intervals for the vertically weighted average when conditioning on $ Y_n = \Phi^{-1}(q) $ for $q$-values between $ 0.05 $ and $ 0.95 $.}
\label{SimTableClassi}
\end{table}

\subsection{Accuracy of unconditional fixed sample confidence intervals}

The accuracy of the unconditional fixed sample confidence interval (\ref{CIMarginal}), in terms of its coverage probability, was investigated for confidence levels between $ 75\% $ and $ 99.9 \% $. The sample size $n$ was selected with levels $ 20, 30, 50, 75 $ and $ 100 $. The levels of the bandwidth parameter for the Gaussian kernel were chosen from $ 0.4 $ to $ 2.0 $. 

Table~\ref{SimTableClassiMarBT1} provides the simulated coverages for a nonparametric bootstrap variance estimator based on $1,000$ replications, for nominal confidence levels between $ 0.75 $ and $ 0.999$. Each entry is based on $ 10,000 $ runs. It can be seen that the proposed confidence interval is surprisingly accurate even for high confidence levels. 

The accuracy can be improved further by using a larger number of bootstrap replications. For comparison, Table~\ref{SimTableClassiMarBT2} provides the results when $ B = 2,500 $ replicates are used to estimate the variance.  

\begin{table}
 \centering\ 
 \begin{tabular}{ccrrrrrrrr} 
 \hline
 $\sigma$ & $n$  & \multicolumn{8}{c}{$1-\alpha$} \\ 
     &      & $ 0.75 $& $ 0.8 $& $ 0.9 $& $ 0.925 $& $ 0.95 $& $ 0.975 $& $ 0.99 $& $ 0.999 $ \\ 
 \hline
$ 0.4 $ & $ 20 $  & $  0.737  $  & $  0.801  $  & $  0.898  $  & $  0.920  $  & $  0.950  $  & $  0.975  $  & $  0.988  $  & $  0.999  $  \\ 
 & $ 30 $  & $  0.752  $  & $  0.803  $  & $  0.903  $  & $  0.925  $  & $  0.948  $  & $  0.975  $  & $  0.988  $  & $  0.999  $  \\ 
 & $ 50 $  & $  0.758  $  & $  0.811  $  & $  0.903  $  & $  0.930  $  & $  0.951  $  & $  0.975  $  & $  0.989  $  & $  0.998  $  \\ 
 & $ 75 $  & $  0.758  $  & $  0.808  $  & $  0.900  $  & $  0.927  $  & $  0.949  $  & $  0.977  $  & $  0.987  $  & $  0.998  $  \\ 
 & $ 100 $  & $  0.760  $  & $  0.800  $  & $  0.904  $  & $  0.924  $  & $  0.955  $  & $  0.974  $  & $  0.989  $  & $  0.998  $  \\ 
 $ 0.6 $ & $ 20 $  & $  0.753  $  & $  0.806  $  & $  0.904  $  & $  0.928  $  & $  0.955  $  & $  0.974  $  & $  0.990  $  & $  0.998  $  \\ 
 & $ 30 $  & $  0.755  $  & $  0.803  $  & $  0.904  $  & $  0.930  $  & $  0.951  $  & $  0.977  $  & $  0.988  $  & $  0.997  $  \\ 
 & $ 50 $  & $  0.765  $  & $  0.816  $  & $  0.908  $  & $  0.929  $  & $  0.956  $  & $  0.977  $  & $  0.990  $  & $  0.998  $  \\ 
 & $ 75 $  & $  0.761  $  & $  0.806  $  & $  0.904  $  & $  0.931  $  & $  0.952  $  & $  0.977  $  & $  0.988  $  & $  0.999  $  \\ 
 & $ 100 $  & $  0.757  $  & $  0.806  $  & $  0.911  $  & $  0.933  $  & $  0.955  $  & $  0.977  $  & $  0.988  $  & $  0.997  $  \\ 
 $ 0.8 $ & $ 20 $  & $  0.759  $  & $  0.804  $  & $  0.906  $  & $  0.933  $  & $  0.955  $  & $  0.975  $  & $  0.991  $  & $  0.998  $  \\ 
 & $ 30 $  & $  0.767  $  & $  0.816  $  & $  0.913  $  & $  0.936  $  & $  0.958  $  & $  0.976  $  & $  0.991  $  & $  0.998  $  \\ 
 & $ 50 $  & $  0.769  $  & $  0.812  $  & $  0.912  $  & $  0.934  $  & $  0.956  $  & $  0.978  $  & $  0.990  $  & $  0.997  $  \\ 
 & $ 75 $  & $  0.767  $  & $  0.809  $  & $  0.908  $  & $  0.932  $  & $  0.955  $  & $  0.976  $  & $  0.990  $  & $  0.998  $  \\ 
 & $ 100 $  & $  0.761  $  & $  0.809  $  & $  0.907  $  & $  0.933  $  & $  0.954  $  & $  0.978  $  & $  0.989  $  & $  0.998  $  \\ 
 $ 1 $ & $ 20 $  & $  0.754  $  & $  0.803  $  & $  0.907  $  & $  0.931  $  & $  0.953  $  & $  0.976  $  & $  0.990  $  & $  0.998  $  \\ 
 & $ 30 $  & $  0.773  $  & $  0.812  $  & $  0.907  $  & $  0.938  $  & $  0.952  $  & $  0.975  $  & $  0.990  $  & $  0.998  $  \\ 
 & $ 50 $  & $  0.771  $  & $  0.809  $  & $  0.912  $  & $  0.933  $  & $  0.957  $  & $  0.978  $  & $  0.990  $  & $  0.998  $  \\ 
 & $ 75 $  & $  0.760  $  & $  0.806  $  & $  0.914  $  & $  0.936  $  & $  0.956  $  & $  0.977  $  & $  0.991  $  & $  0.998  $  \\ 
 & $ 100 $  & $  0.764  $  & $  0.810  $  & $  0.912  $  & $  0.934  $  & $  0.956  $  & $  0.975  $  & $  0.989  $  & $  0.998  $  \\ 
 $ 1.2 $ & $ 20 $  & $  0.754  $  & $  0.805  $  & $  0.910  $  & $  0.921  $  & $  0.958  $  & $  0.976  $  & $  0.991  $  & $  0.999  $  \\ 
 & $ 30 $  & $  0.761  $  & $  0.809  $  & $  0.913  $  & $  0.930  $  & $  0.954  $  & $  0.978  $  & $  0.991  $  & $  0.999  $  \\ 
 & $ 50 $  & $  0.764  $  & $  0.819  $  & $  0.908  $  & $  0.937  $  & $  0.957  $  & $  0.980  $  & $  0.991  $  & $  0.999  $  \\ 
 & $ 75 $  & $  0.759  $  & $  0.808  $  & $  0.906  $  & $  0.933  $  & $  0.955  $  & $  0.979  $  & $  0.991  $  & $  0.999  $  \\ 
 & $ 100 $  & $  0.756  $  & $  0.806  $  & $  0.909  $  & $  0.932  $  & $  0.956  $  & $  0.979  $  & $  0.992  $  & $  0.998  $  \\ 
 $ 1.4 $ & $ 20 $  & $  0.742  $  & $  0.804  $  & $  0.897  $  & $  0.922  $  & $  0.951  $  & $  0.973  $  & $  0.988  $  & $  0.998  $  \\ 
 & $ 30 $  & $  0.755  $  & $  0.801  $  & $  0.905  $  & $  0.930  $  & $  0.957  $  & $  0.979  $  & $  0.994  $  & $  1.000  $  \\ 
 & $ 50 $  & $  0.764  $  & $  0.809  $  & $  0.907  $  & $  0.934  $  & $  0.955  $  & $  0.979  $  & $  0.991  $  & $  0.999  $  \\ 
 & $ 75 $  & $  0.757  $  & $  0.811  $  & $  0.908  $  & $  0.934  $  & $  0.954  $  & $  0.980  $  & $  0.991  $  & $  0.999  $  \\ 
 & $ 100 $  & $  0.760  $  & $  0.802  $  & $  0.910  $  & $  0.927  $  & $  0.956  $  & $  0.977  $  & $  0.990  $  & $  0.999  $  \\ 
 $ 1.6 $ & $ 20 $  & $  0.748  $  & $  0.794  $  & $  0.896  $  & $  0.925  $  & $  0.950  $  & $  0.971  $  & $  0.987  $  & $  0.998  $  \\ 
 & $ 30 $  & $  0.754  $  & $  0.802  $  & $  0.899  $  & $  0.929  $  & $  0.953  $  & $  0.977  $  & $  0.989  $  & $  0.999  $  \\ 
 & $ 50 $  & $  0.758  $  & $  0.801  $  & $  0.909  $  & $  0.930  $  & $  0.954  $  & $  0.978  $  & $  0.992  $  & $  0.999  $  \\ 
 & $ 75 $  & $  0.753  $  & $  0.806  $  & $  0.901  $  & $  0.931  $  & $  0.955  $  & $  0.978  $  & $  0.990  $  & $  0.999  $  \\ 
 & $ 100 $  & $  0.764  $  & $  0.807  $  & $  0.904  $  & $  0.928  $  & $  0.953  $  & $  0.976  $  & $  0.991  $  & $  0.999  $  \\ 
 $ 2 $ & $ 20 $  & $  0.733  $  & $  0.789  $  & $  0.888  $  & $  0.911  $  & $  0.942  $  & $  0.968  $  & $  0.985  $  & $  0.997  $  \\ 
 & $ 30 $  & $  0.744  $  & $  0.795  $  & $  0.898  $  & $  0.921  $  & $  0.950  $  & $  0.975  $  & $  0.989  $  & $  0.999  $  \\ 
 & $ 50 $  & $  0.750  $  & $  0.801  $  & $  0.900  $  & $  0.929  $  & $  0.950  $  & $  0.977  $  & $  0.992  $  & $  0.999  $  \\ 
 & $ 75 $  & $  0.752  $  & $  0.802  $  & $  0.901  $  & $  0.925  $  & $  0.953  $  & $  0.975  $  & $  0.992  $  & $  0.999  $  \\ 
 & $ 100 $  & $  0.751  $  & $  0.799  $  & $  0.905  $  & $  0.927  $  & $  0.952  $  & $  0.980  $  & $  0.989  $  & $  0.999  $  \\ 
 \hline
\end{tabular}

 \ \\[0.1cm]
\caption{Coverage of unconditional fixed sample confidence intervals with bootstrap variance estimate based on $ B = 1,000 $ replications for the vertically weighted average,  for nominal confidence levels between $ 0.75 $ and $ 0.999$.}
\label{SimTableClassiMarBT1}
\end{table}

\begin{table}
 \centering\ 
 \begin{tabular}{ccrrrrrrrr} 
 \hline
 $\sigma$ & $n$ & \multicolumn{8}{c}{$1-\alpha$} \\
  &  & $ 0.75 $& $ 0.8 $& $ 0.9 $& $ 0.925 $& $ 0.95 $& $ 0.975 $& $ 0.99 $& $ 0.999 $ \\ 
 \hline
$ 0.4 $ & $ 20 $  & $  0.747  $  & $  0.793  $  & $  0.900  $  & $  0.930  $  & $  0.946  $  & $  0.972  $  & $  0.990  $  & $  0.999  $  \\ 
 & $ 30 $  & $  0.749  $  & $  0.802  $  & $  0.901  $  & $  0.929  $  & $  0.950  $  & $  0.974  $  & $  0.986  $  & $  0.998  $  \\ 
 & $ 50 $  & $  0.756  $  & $  0.803  $  & $  0.908  $  & $  0.929  $  & $  0.954  $  & $  0.977  $  & $  0.987  $  & $  0.999  $  \\ 
 & $ 75 $  & $  0.755  $  & $  0.807  $  & $  0.903  $  & $  0.930  $  & $  0.951  $  & $  0.976  $  & $  0.988  $  & $  0.998  $  \\ 
 & $ 100 $  & $  0.766  $  & $  0.806  $  & $  0.902  $  & $  0.928  $  & $  0.951  $  & $  0.976  $  & $  0.988  $  & $  0.998  $  \\ 
 $ 0.6 $ & $ 20 $  & $  0.758  $  & $  0.808  $  & $  0.904  $  & $  0.928  $  & $  0.956  $  & $  0.975  $  & $  0.989  $  & $  0.999  $  \\ 
 & $ 30 $  & $  0.763  $  & $  0.813  $  & $  0.902  $  & $  0.929  $  & $  0.954  $  & $  0.974  $  & $  0.989  $  & $  0.997  $  \\ 
 & $ 50 $  & $  0.765  $  & $  0.812  $  & $  0.909  $  & $  0.933  $  & $  0.955  $  & $  0.975  $  & $  0.988  $  & $  0.998  $  \\ 
 & $ 75 $  & $  0.758  $  & $  0.816  $  & $  0.910  $  & $  0.932  $  & $  0.952  $  & $  0.977  $  & $  0.988  $  & $  0.997  $  \\ 
 & $ 100 $  & $  0.766  $  & $  0.813  $  & $  0.912  $  & $  0.932  $  & $  0.953  $  & $  0.976  $  & $  0.989  $  & $  0.998  $  \\ 
 $ 0.8 $ & $ 20 $  & $  0.764  $  & $  0.811  $  & $  0.907  $  & $  0.931  $  & $  0.953  $  & $  0.976  $  & $  0.989  $  & $  0.999  $  \\ 
 & $ 30 $  & $  0.764  $  & $  0.816  $  & $  0.909  $  & $  0.937  $  & $  0.958  $  & $  0.976  $  & $  0.989  $  & $  0.998  $  \\ 
 & $ 50 $  & $  0.766  $  & $  0.815  $  & $  0.911  $  & $  0.934  $  & $  0.958  $  & $  0.980  $  & $  0.990  $  & $  0.998  $  \\ 
 & $ 75 $  & $  0.767  $  & $  0.818  $  & $  0.908  $  & $  0.931  $  & $  0.953  $  & $  0.975  $  & $  0.989  $  & $  0.998  $  \\ 
 & $ 100 $  & $  0.756  $  & $  0.815  $  & $  0.908  $  & $  0.932  $  & $  0.955  $  & $  0.978  $  & $  0.990  $  & $  0.998  $  \\ 
 $ 1 $ & $ 20 $  & $  0.759  $  & $  0.810  $  & $  0.909  $  & $  0.933  $  & $  0.954  $  & $  0.977  $  & $  0.991  $  & $  0.998  $  \\ 
 & $ 30 $  & $  0.763  $  & $  0.827  $  & $  0.914  $  & $  0.939  $  & $  0.956  $  & $  0.979  $  & $  0.990  $  & $  0.998  $  \\ 
 & $ 50 $  & $  0.770  $  & $  0.815  $  & $  0.915  $  & $  0.938  $  & $  0.959  $  & $  0.978  $  & $  0.992  $  & $  0.998  $  \\ 
 & $ 75 $  & $  0.764  $  & $  0.815  $  & $  0.909  $  & $  0.928  $  & $  0.958  $  & $  0.979  $  & $  0.989  $  & $  0.997  $  \\ 
 & $ 100 $  & $  0.762  $  & $  0.811  $  & $  0.906  $  & $  0.932  $  & $  0.957  $  & $  0.978  $  & $  0.991  $  & $  0.998  $  \\ 
 $ 1.2 $ & $ 20 $  & $  0.755  $  & $  0.797  $  & $  0.909  $  & $  0.926  $  & $  0.956  $  & $  0.979  $  & $  0.989  $  & $  0.999  $  \\ 
 & $ 30 $  & $  0.759  $  & $  0.810  $  & $  0.907  $  & $  0.933  $  & $  0.957  $  & $  0.978  $  & $  0.991  $  & $  0.999  $  \\ 
 & $ 50 $  & $  0.758  $  & $  0.807  $  & $  0.913  $  & $  0.928  $  & $  0.956  $  & $  0.978  $  & $  0.991  $  & $  0.999  $  \\ 
 & $ 75 $  & $  0.759  $  & $  0.806  $  & $  0.909  $  & $  0.934  $  & $  0.956  $  & $  0.978  $  & $  0.989  $  & $  0.999  $  \\ 
 & $ 100 $  & $  0.754  $  & $  0.806  $  & $  0.910  $  & $  0.933  $  & $  0.955  $  & $  0.979  $  & $  0.991  $  & $  0.999  $  \\ 
 $ 1.4 $ & $ 20 $  & $  0.749  $  & $  0.795  $  & $  0.898  $  & $  0.927  $  & $  0.951  $  & $  0.977  $  & $  0.988  $  & $  0.998  $  \\ 
 & $ 30 $  & $  0.757  $  & $  0.803  $  & $  0.907  $  & $  0.935  $  & $  0.954  $  & $  0.978  $  & $  0.990  $  & $  0.998  $  \\ 
 & $ 50 $  & $  0.763  $  & $  0.810  $  & $  0.907  $  & $  0.931  $  & $  0.957  $  & $  0.976  $  & $  0.990  $  & $  0.999  $  \\ 
 & $ 75 $  & $  0.758  $  & $  0.810  $  & $  0.906  $  & $  0.931  $  & $  0.956  $  & $  0.978  $  & $  0.990  $  & $  0.998  $  \\ 
 & $ 100 $  & $  0.754  $  & $  0.811  $  & $  0.908  $  & $  0.934  $  & $  0.955  $  & $  0.978  $  & $  0.991  $  & $  0.998  $  \\ 
 $ 1.6 $ & $ 20 $  & $  0.745  $  & $  0.801  $  & $  0.902  $  & $  0.920  $  & $  0.946  $  & $  0.972  $  & $  0.986  $  & $  0.998  $  \\ 
 & $ 30 $  & $  0.741  $  & $  0.804  $  & $  0.906  $  & $  0.932  $  & $  0.950  $  & $  0.976  $  & $  0.991  $  & $  0.999  $  \\ 
 & $ 50 $  & $  0.753  $  & $  0.812  $  & $  0.905  $  & $  0.936  $  & $  0.952  $  & $  0.975  $  & $  0.991  $  & $  0.999  $  \\ 
 & $ 75 $  & $  0.747  $  & $  0.800  $  & $  0.904  $  & $  0.930  $  & $  0.954  $  & $  0.977  $  & $  0.991  $  & $  0.999  $  \\ 
 & $ 100 $  & $  0.755  $  & $  0.808  $  & $  0.906  $  & $  0.936  $  & $  0.952  $  & $  0.977  $  & $  0.992  $  & $  0.999  $  \\ 
 $ 2 $ & $ 20 $  & $  0.742  $  & $  0.787  $  & $  0.892  $  & $  0.918  $  & $  0.948  $  & $  0.966  $  & $  0.986  $  & $  0.996  $  \\ 
 & $ 30 $  & $  0.740  $  & $  0.793  $  & $  0.898  $  & $  0.924  $  & $  0.946  $  & $  0.973  $  & $  0.989  $  & $  0.999  $  \\ 
 & $ 50 $  & $  0.750  $  & $  0.799  $  & $  0.902  $  & $  0.926  $  & $  0.950  $  & $  0.976  $  & $  0.991  $  & $  0.999  $  \\ 
 & $ 75 $  & $  0.750  $  & $  0.812  $  & $  0.903  $  & $  0.929  $  & $  0.955  $  & $  0.978  $  & $  0.992  $  & $  0.999  $  \\ 
 & $ 100 $  & $  0.758  $  & $  0.801  $  & $  0.893  $  & $  0.935  $  & $  0.950  $  & $  0.977  $  & $  0.992  $  & $  0.999  $  \\ 
 \hline
\end{tabular}

 \ \\[0.1cm]
\caption{Coverage of unconditional fixed sample confidence intervals with bootstrap variance estimate based on $ B = 2,500 $ replications for the vertically weighted average,  for nominal confidence levels between $ 0.75 $ and $ 0.999$.}
\label{SimTableClassiMarBT2}
\end{table}

\subsection{Accuracy of fixed width confidence intervals: Fixed initial sample sizes}

The next simulation studies the coverage probability of the fixed width confidence intervals for a selection of quantiles of the distribution and for different confidence levels and values for the size of the initial sample size. Recall that now the final sample size is random and denoted by $N$.

In a first set of simulations, we investigate the accuracy when the initial sample size is fixed at certain levels, which is sometimes necessary in practice. We used the levels $ n_0 = 20, 30 $ and $ 50 $. Table~\ref{SimTableAs1} shows that the coverage probabilities for values of $ Y_N = \Phi^{-1}(q) $ for $q $ taken from the set
$ \{ 0.05, 0.1, \dots, 0.9, 0.95 \} $. The precision was specified as $ d = 0.2 $. The confidence level  was chosen as  $ 0.9, 0.95 $ and $ 0.975 $.
For each case the upper entry provides the simulated coverage probability and the lower entry the
simulated mean sample size. It turns out that the coverage is fairly good even for small
initial sample sizes in the center of the distribution, but it decreases in the tails. Observe that in various settings
the coverage probabilities exceed the nominal value. This happens when the first stage sample is larger or equal to the required sample sizes, as can be seen from the simulated mean sample sizes.

The accuracy improves for higher precision as shown by the results provided in Table~\ref{SimTableAs2} for $ d = 0.1 $. Now larger sample sizes are needed and one can see that the two stage approach performs quite good, if $n_0$ is not too small. For $ n_0 = 50 $ the coverage is only slightly smaller than the nominal value in the center of the distribution, but it decreases in the tails.

\begin{table}
 \centering\ 
 \begin{tabular}{crrrrrrrr} 
 \hline
 $1-\alpha$ & $n_0$ & \multicolumn{7}{c}{$q$} \\ 
  &  & $ 0.05 $& $ 0.1 $& $ 0.3 $& $ 0.5 $& $ 0.8 $& $ 0.9 $& $ 0.95 $\\ 
 \hline
$ 0.9 $  & $ 20 $ & $  0.87  $  & $  0.89  $  & $  0.94  $  & $  0.95  $  & $  0.92  $  & $  0.89  $  & $  0.87  $  \\ 
  &  & $  29.29  $  & $  24.85  $  & $  20.63  $  & $  20.28  $  & $  21.64  $  & $  24.95  $  & $  29.09  $  \\ 
  & $ 30 $ & $  0.89  $  & $  0.93  $  & $  0.98  $  & $  0.98  $  & $  0.97  $  & $  0.93  $  & $  0.89  $  \\ 
  &  & $  37.22  $  & $  32.04  $  & $  30.02  $  & $  30.00  $  & $  30.21  $  & $  32.00  $  & $  37.59  $  \\ 
  & $ 50 $ & $  0.94  $  & $  0.98  $  & $  1.00  $  & $  1.00  $  & $  0.99  $  & $  0.98  $  & $  0.94  $  \\ 
  &  & $  52.11  $  & $  50.08  $  & $  50.00  $  & $  50.00  $  & $  50.00  $  & $  50.08  $  & $  52.09  $  \\ 
 $ 0.95 $  & $ 20 $ & $  0.91  $  & $  0.91  $  & $  0.95  $  & $  0.95  $  & $  0.94  $  & $  0.91  $  & $  0.91  $  \\ 
  &  & $  37.42  $  & $  31.64  $  & $  23.27  $  & $  22.03  $  & $  25.71  $  & $  31.90  $  & $  37.65  $  \\ 
  & $ 30 $ & $  0.92  $  & $  0.94  $  & $  0.98  $  & $  0.98  $  & $  0.97  $  & $  0.94  $  & $  0.92  $  \\ 
  &  & $  47.25  $  & $  37.67  $  & $  30.56  $  & $  30.17  $  & $  31.69  $  & $  37.62  $  & $  47.15  $  \\ 
  & $ 50 $ & $  0.95  $  & $  0.97  $  & $  1.00  $  & $  1.00  $  & $  0.99  $  & $  0.98  $  & $  0.95  $  \\ 
  &  & $  59.15  $  & $  51.19  $  & $  50.00  $  & $  50.00  $  & $  50.02  $  & $  51.19  $  & $  58.82  $  \\ 
 $ 0.975 $  & $ 20 $ & $  0.93  $  & $  0.93  $  & $  0.96  $  & $  0.96  $  & $  0.95  $  & $  0.93  $  & $  0.93  $  \\ 
  &  & $  46.50  $  & $  39.79  $  & $  28.07  $  & $  25.86  $  & $  31.64  $  & $  39.93  $  & $  47.01  $  \\ 
  & $ 30 $ & $  0.94  $  & $  0.95  $  & $  0.98  $  & $  0.98  $  & $  0.97  $  & $  0.95  $  & $  0.94  $  \\ 
  &  & $  59.67  $  & $  45.86  $  & $  32.69  $  & $  31.40  $  & $  35.80  $  & $  45.89  $  & $  59.12  $  \\ 
  & $ 50 $ & $  0.96  $  & $  0.98  $  & $  1.00  $  & $  1.00  $  & $  0.99  $  & $  0.98  $  & $  0.95  $  \\ 
  &  & $  70.44  $  & $  55.46  $  & $  50.03  $  & $  50.00  $  & $  50.35  $  & $  55.27  $  & $  70.63  $  \\ 
 \hline
\end{tabular}

 \ \\[0.1cm]
\caption{For precision $ d = 0.2 $: Coverage probabilities (first row) and mean sample sizes (second row) for confidence levels
$ 90\%, 95\% $ and $ 97.5\% $ and intial sample sizes $ 20 $ and $ 30 $, across the distribution, i.e. when conditioning on $ Y_n = \Phi^{-1}(q) $ for $ q $ between $ 0.05 $ and $ 0.95 $.}
\label{SimTableAs1}
\end{table}

\begin{table}
 \centering\ 
 \begin{tabular}{crrrrrrrr} 
 \hline
 $1-\alpha$ & $n_0$ & \multicolumn{7}{c}{$q$} \\ 
  &  & $ 0.05 $& $ 0.1 $& $ 0.3 $& $ 0.5 $& $ 0.8 $& $ 0.9 $& $ 0.95 $\\ 
 \hline
$ 0.9 $  & $ 20 $ & $  0.74  $  & $  0.79  $  & $  0.86  $  & $  0.86  $  & $  0.85  $  & $  0.79  $  & $  0.73  $  \\ 
  &  & $  95.42  $  & $  82.06  $  & $  57.16  $  & $  52.21  $  & $  65.15  $  & $  82.89  $  & $  95.25  $  \\ 
  & $ 30 $ & $  0.79  $  & $  0.84  $  & $  0.88  $  & $  0.88  $  & $  0.87  $  & $  0.84  $  & $  0.78  $  \\ 
  &  & $  121.66  $  & $  93.90  $  & $  61.31  $  & $  55.62  $  & $  70.98  $  & $  94.90  $  & $  121.35  $  \\ 
  & $ 50 $ & $  0.85  $  & $  0.88  $  & $  0.90  $  & $  0.90  $  & $  0.89  $  & $  0.87  $  & $  0.85  $  \\ 
  &  & $  143.60  $  & $  104.03  $  & $  65.27  $  & $  59.56  $  & $  75.93  $  & $  104.11  $  & $  144.14  $  \\ 
 $ 0.95 $  & $ 20 $ & $  0.80  $  & $  0.86  $  & $  0.91  $  & $  0.92  $  & $  0.90  $  & $  0.86  $  & $  0.79  $  \\ 
  &  & $  135.03  $  & $  116.80  $  & $  80.23  $  & $  73.68  $  & $  91.78  $  & $  116.21  $  & $  135.37  $  \\ 
  & $ 30 $ & $  0.85  $  & $  0.90  $  & $  0.93  $  & $  0.93  $  & $  0.92  $  & $  0.90  $  & $  0.85  $  \\ 
  &  & $  171.55  $  & $  133.33  $  & $  86.22  $  & $  78.04  $  & $  99.53  $  & $  133.41  $  & $  173.50  $  \\ 
  & $ 50 $ & $  0.90  $  & $  0.93  $  & $  0.94  $  & $  0.94  $  & $  0.94  $  & $  0.93  $  & $  0.91  $  \\ 
  &  & $  202.79  $  & $  147.38  $  & $  91.00  $  & $  82.39  $  & $  106.34  $  & $  146.49  $  & $  203.86  $  \\ 
 $ 0.975 $  & $ 20 $ & $  0.83  $  & $  0.90  $  & $  0.94  $  & $  0.95  $  & $  0.94  $  & $  0.89  $  & $  0.83  $  \\ 
  &  & $  175.23  $  & $  151.46  $  & $  104.28  $  & $  95.71  $  & $  119.26  $  & $  152.51  $  & $  175.19  $  \\ 
  & $ 30 $ & $  0.88  $  & $  0.94  $  & $  0.96  $  & $  0.96  $  & $  0.95  $  & $  0.93  $  & $  0.88  $  \\ 
  &  & $  222.90  $  & $  174.25  $  & $  112.90  $  & $  102.16  $  & $  130.12  $  & $  175.27  $  & $  224.09  $  \\ 
  & $ 50 $ & $  0.94  $  & $  0.96  $  & $  0.96  $  & $  0.97  $  & $  0.97  $  & $  0.96  $  & $  0.94  $  \\ 
  &  & $  263.76  $  & $  191.09  $  & $  118.59  $  & $  107.39  $  & $  138.95  $  & $  191.04  $  & $  265.02  $  \\ 
 \hline
\end{tabular}

 \ \\[0.1cm]
\caption{For precision $d=0.1$: Coverage probabilities (first row) and mean sample sizes (second row) for confidence levels
$ 90\%, 95\% $ and $ 97.5\% $ and intial sample sizes $ 20 $ and $ 30 $, across the distribution,  i.e. when conditioning on $ Y_N = \Phi^{-1}(q) $ for $ q $ between $ 0.05 $ and $ 0.95 $.}
\label{SimTableAs2}
\end{table}

\subsection{Accuracy of fixed width confidence intervals: Two stage procedure}

Let us now analyze the accuracy of the proposed two stage procedure with initial sample size $ n_0 $ selected using the
rule (\ref{N0s}). The bandwidth parameter $ \sigma = 0.6 $ was used. Table~\ref{SimTableAsN02} provides the corresponding results when the fixed width interval is given by $ d = 0.2 $, where each entry is based on $ 50,000 $ independent runs. We see that the final sample sizes required for a confidence level of $ 90\% $ are surprisingly small and, on average, range between $ 18 $ in the center of the distribution and $ 26 $ around the $ 5\% $ and $ 95\% $ percentile points. For a large confidence level of $ 97.5\% $, the sample sizes range between $ 29 $ and $ 56 $. The coverage is, for $ 90\% $ confidence, surprisingly good with a higher than nominal coverage in the center and a lower than nominal coverage in the tails. The picture is quite similar for larger values of the confidence level.

\begin{table}
 \centering\ 
 \begin{tabular}{crrrrrrr} 
 \hline
 $1-\alpha$ & \multicolumn{7}{c}{$q$} \\ 
 & $ 0.05 $& $ 0.1 $& $ 0.3 $& $ 0.5 $& $ 0.8 $& $ 0.9 $& $ 0.95 $\\ 
 \hline
$ 0.9 $  & $  0.862  $  & $  0.873  $  & $  0.927  $  & $  0.938  $  & $  0.911  $  & $  0.870  $  & $  0.861  $  \\ 
  & $  25.69  $  & $  22.77  $  & $  18.22  $  & $  17.70  $  & $  19.43  $  & $  22.72  $  & $  25.59  $  \\ 
 $ 0.95 $  & $  0.909  $  & $  0.915  $  & $  0.954  $  & $  0.961  $  & $  0.942  $  & $  0.916  $  & $  0.907  $  \\ 
  & $  38.70  $  & $  32.54  $  & $  23.85  $  & $  22.74  $  & $  26.21  $  & $  32.50  $  & $  38.60  $  \\ 
 $ 0.975 $  & $  0.934  $  & $  0.946  $  & $  0.973  $  & $  0.978  $  & $  0.969  $  & $  0.943  $  & $  0.933  $  \\ 
  & $  55.08  $  & $  43.94  $  & $  30.63  $  & $  28.82  $  & $  34.21  $  & $  43.92  $  & $  55.22  $  \\ 
 \hline
\end{tabular}

 \ \\[0.1cm]
\caption{For precision $d=0.2$: Coverage probabilities (first row) and mean sample sizes (second row) for confidence levels
$ 90\%, 95\% $ and $ 97.5\% $ using the proposed rule (\ref{N0s}), when conditioning on $ Y_n = \Phi^{-1}(q) $ for $q$ between $ 0.05 $ and $ 0.95 $.}
\label{SimTableAsN02}
\end{table}

When the precision is increased to, say, $ d = 0.1 $, Table~\ref{SimTableAsN01} shows that the required final sample sizes increase and range for a confidence level of $90\% $ between $ 58 $ and $ 138 $, whereas for $ 95\% $ confidence they increase to $ 83 $ and $ 206 $, respectively. The accuracy of the coverage probability is considerably better across all considered $y$--values. In the center of the distribution the coverage is very good and slightly below the nominal coverage. It decreases in the tails where the true coverage is somewhat below the nominal value.

\begin{table}
 \centering\ 
 \begin{tabular}{crrrrrrr} 
 \hline
 $1-\alpha$ & \multicolumn{7}{c}{$q$} \\ 
 & $ 0.05 $& $ 0.1 $& $ 0.3 $& $ 0.5 $& $ 0.8 $& $ 0.9 $& $ 0.95 $\\ 
 \hline
$ 0.9 $  & $  0.826  $  & $  0.862  $  & $  0.888  $  & $  0.892  $  & $  0.882  $  & $  0.863  $  & $  0.834  $  \\ 
  & $  137.50  $  & $  101.86  $  & $  63.73  $  & $  57.93  $  & $  74.14  $  & $  101.34  $  & $  138.51  $  \\ 
 $ 0.95 $  & $  0.906  $  & $  0.926  $  & $  0.940  $  & $  0.943  $  & $  0.938  $  & $  0.926  $  & $  0.906  $  \\ 
  & $  205.95  $  & $  147.69  $  & $  91.35  $  & $  82.53  $  & $  106.92  $  & $  147.80  $  & $  206.01  $  \\ 
 $ 0.975 $  & $  0.946  $  & $  0.961  $  & $  0.970  $  & $  0.970  $  & $  0.967  $  & $  0.962  $  & $  0.946  $  \\ 
  & $  279.64  $  & $  197.20  $  & $  120.52  $  & $  108.79  $  & $  141.32  $  & $  197.11  $  & $  278.59  $  \\ 
 \hline
\end{tabular}

 \ \\[0.1cm]
\caption{For precision $d=0.1$: Coverage probabilities (first row) and mean sample sizes (second row) for confidence levels
$ 90\%, 95\% $ and $ 97.5\% $ using the proposed rule (\ref{N0s}), when conditioning on $ Y_n = \Phi^{-1}(q) $ for $q$ between $ 0.05 $ and $ 0.95 $.}
\label{SimTableAsN01}
\end{table}

\subsection{Accuracy of bootstrapped fixed width confidence intervals}

To improve upon the normal approximation, we used the smooth bootstrap two stage procedure proposed in Section~\ref{Sec: Fixed-width CI}. The bootstrap was based on $ B = 2,000 $ repetitions and each case was simulated using $ 50,000$  runs.

For $ d = 0.2 $, by comparing Tables~\ref{SimTableAsN02} and Table~\ref{SimTableBTN02}, we can see that there is slight improvement, but not in all cases. A plausible explanation is that the sample sizes are quite small. 
 If $ d = 0.1 $, however, it becomes apparent that the
bootstrap improves upon the approximation based on the central limit theorem, as can be seen from Table~\ref{SimTableBTN01}: Uniformly across all settings studied here, the simulated coverage probabilities of the bootstrapped fixed width intervals are closer to their nominal values than the corresponding coverages of the asymptotic fixed width intervals. 

\begin{table}
 \centering\ 
 \begin{tabular}{crrrrrrr} 
 \hline
 $1-\alpha$ & \multicolumn{7}{c}{$q$} \\ 
 & $ 0.05 $& $ 0.1 $& $ 0.3 $& $ 0.5 $& $ 0.8 $& $ 0.9 $& $ 0.95 $\\ 
 \hline
$ 0.9 $  & $  0.85  $  & $  0.88  $  & $  0.92  $  & $  0.93  $  & $  0.91  $  & $  0.87  $  & $  0.85  $  \\ 
  & $  21.28  $  & $  20.36  $  & $  17.67  $  & $  17.35  $  & $  18.46  $  & $  20.34  $  & $  21.34  $  \\ 
 $ 0.95 $  & $  0.91  $  & $  0.92  $  & $  0.96  $  & $  0.96  $  & $  0.94  $  & $  0.92  $  & $  0.91  $  \\ 
  & $  32.92  $  & $  30.36  $  & $  23.74  $  & $  22.68  $  & $  25.76  $  & $  30.25  $  & $  33.04  $  \\ 
 $ 0.975 $  & $  0.94  $  & $  0.95  $  & $  0.97  $  & $  0.98  $  & $  0.97  $  & $  0.95  $  & $  0.93  $  \\ 
  & $  48.65  $  & $  42.81  $  & $  31.29  $  & $  29.42  $  & $  34.79  $  & $  42.64  $  & $  48.39  $  \\ 
 \hline
\end{tabular}

 \ \\[0.1cm]
\caption{Bootstrapped two stage procedure for precision $d=0.2$: Coverage probabilities (first row) and mean sample sizes (second row), when conditioning on $ Y_n = \Phi^{-1}(q) $ for $q$ between $ 0.05 $ and $ 0.95 $.}
\label{SimTableBTN02}
\end{table}

\begin{table}
 \centering\ 
 \begin{tabular}{crrrrrrr} 
 \hline
 $1-\alpha$ & \multicolumn{7}{c}{$q$} \\ 
 & $ 0.05 $& $ 0.1 $& $ 0.3 $& $ 0.5 $& $ 0.8 $& $ 0.9 $& $ 0.95 $\\ 
 \hline
$ 0.9 $  & $  0.84  $  & $  0.87  $  & $  0.90  $  & $  0.90  $  & $  0.89  $  & $  0.88  $  & $  0.85  $  \\ 
  & $  128.73  $  & $  100.32  $  & $  67.34  $  & $  61.68  $  & $  76.79  $  & $  99.91  $  & $  128.87  $  \\ 
 $ 0.95 $  & $  0.92  $  & $  0.93  $  & $  0.95  $  & $  0.96  $  & $  0.95  $  & $  0.94  $  & $  0.92  $  \\ 
  & $  195.46  $  & $  147.47  $  & $  96.87  $  & $  88.73  $  & $  111.36  $  & $  147.22  $  & $  195.38  $  \\ 
 $ 0.975 $  & $  0.96  $  & $  0.97  $  & $  0.97  $  & $  0.98  $  & $  0.97  $  & $  0.97  $  & $  0.96  $  \\ 
  & $  268.17  $  & $  198.31  $  & $  128.25  $  & $  116.65  $  & $  147.72  $  & $  197.92  $  & $  267.88  $  \\ 
 \hline
\end{tabular}

 \ \\[0.1cm]
\caption{Bootstrapped two stage procedure for precision $d=0.1$: Coverage probabilities (first row) and mean sample sizes (second row), , when conditioning on $ Y_n = \Phi^{-1}(q) $ for $q$ between $ 0.05 $ and $ 0.95 $.}
\label{SimTableBTN01}
\end{table}

\section{Conclusions}

The vertically weighted average represents an interesting jump-preserving signal estimator which denoises without corrupting finer details of a signal. Since the effective number of observations used by the estimator to calculate the signal estimate depends on the location of the current observation, the assessment of the precision of this estimator, both in terms of variance estimation and construction of confidence intervals, has been a delicate open problem. 

Recent theoretical results about the asymptotic behavior of the vertically weighted average allow the consideration of confidence intervals based on that jump-preserving signal estimator. We discuss two common statistical approaches to construct confidence intervals: First, asymptotic fixed sample confidence intervals relying on the asymptotic normality of the appropriately standardized estimator, leading to confidence intervals of a random width but based on a sample of fixed sample size. Second, fixed width confidence intervals where one specifies the precision (i.e. width of the interval) in advance. Here the sample size has to be estimated and hence becomes random. We propose to rely on a two stage procedure, where the first stage sample (of a fixed initial sample size) is used to estimate the standard error of the vertically weighted average, in order to calculate an estimator of the required final sample size for the second stage. At the second stage the confidence interval is then calculated based on the vertically weighted average using the final sample, which consists of the initial sample and, if needed, additional observations. 

To estimate the dispersion of the vertically weighted estimator, we propose to use jackknife variance estimation and the bootstrap, respectively. The bootstrap provides convincing results, but it is more demanding from a computational point of view. 

Our simulations indicate that unconditional fixed sample confidence intervales based on bootstrap variance estimation perform well. For fixed width confidence intervals using jackknife variance estimation, the coverage is generally accurate except when the current observation is located in the tails of the underlying distribution, which results in lower than nominal coverage probabilities. The accuracy can be further improved by the bootstrap for high precision intervals. 




\end{document}